\documentclass{ifacconf}
\usepackage{algorithm}
\usepackage{algpseudocode}
\usepackage{amsmath} 
\usepackage{amssymb} 
\usepackage{natbib}        
\setcitestyle{square}

\newtheorem{theorem}{Theorem}      
         
\newtheorem{proposition}{Proposition}
  
\newtheorem{definition}{Definition}

\usepackage{graphicx}      
\usepackage{natbib}        
\begin{document}
\begin{frontmatter}

\title{Learning, Misspecification, and Cognitive Arbitrage in Linear-Quadratic Network Games} 

\author[NYU]{Quanyan Zhu}
\author[NYU]{Zhengye Han} 

\address[NYU]{Department of Electrical and Computer Engineering, \\
              New York University, Brooklyn, NY 11201 USA \\
              (e-mail: qz494@nyu.edu, zh3286@nyu.edu).}

\begin{abstract}
We study strategic interaction in linear-quadratic network games where agents act on subjective, misspecified models of their environment. Agents observe noisy aggregate signals generated by local network externalities and interpret them through simplified conjectures, such as constant or mean-field representations. We characterize the long-run behavior using the Berk-Nash equilibrium (BNE) concept, establishing conditions under which BNE diverges from the Nash equilibrium of the perfectly specified game. We quantify this divergence using a Value of Misspecification (VoM) metric. Building on this framework, we introduce \textit{cognitive arbitrage}—a design paradigm where a system designer strategically shapes agents' conjectures via minimal observation distortions to steer equilibrium outcomes. We formulate the cognitive arbitrage problem as a Stackelberg optimization with closed-form solutions and prove the convergence of a two-time-scale learning algorithm to the optimal BNE. Our results provide a principled framework for influencing behavior in networked systems with bounded rationality, offering a new perspective on mechanism design that operates on agents' representations rather than their incentives.
\end{abstract}


\end{frontmatter}

\section{Introduction}
Network games arise naturally in a wide range of modern engineered and socio-technical systems, including federated learning, social and economic networks, distributed control systems, and emerging agentic AI networks \citep{Federated}. In such settings, multiple decision-makers interact through structured interdependencies, where each agent's payoff depends not only on her own action but also on the actions of others connected through a network. Understanding how agents form decisions and equilibria in these environments is therefore a central problem in networked decision-making and multi-agent systems \citep{jackson2008social, ballester2006s}.

A key challenge in network games is that agents rarely possess an objective or complete view of the underlying environment. Instead, agents operate based on subjective internal models of the world, constructed from partial observations, limited communication, and bounded cognitive or computational resources \citep{Behavioral}. Decisions are made relative to these subjective models, and agents are typically satisfied as long as their internal models are consistent with the observations generated by the real world, rather than objectively correct. This phenomenon is particularly pronounced in large-scale networks, where the true interaction structure may be complex, heterogeneous, nonlinear, and time-varying. In practice, agents often resort to simplified representations, such as linear or low-dimensional models, to approximate and reason about a much more complex environment \citep{bertsekas1995neuro}.

This perspective is closely connected to, but distinct from, the classical theory of mean-field games. Most existing work in mean-field games assumes that agents observe the mean field objectively, either directly or through a common aggregate signal \citep{lasry2007mean, huang2006large}. However, this assumption is often unrealistic. In many networked systems, agents only observe the actions of a local neighborhood and must infer or conjecture the aggregate effect of the population. As a result, agents form subjective mean-field representations based on local information, which may differ systematically from the true population average \citep{caines2017mean}. These subjective mean fields, rather than the true mean field, ultimately drive agents' decision-making processes. Importantly, mean-field-based decisions can still be acceptable---and even stable---provided that the distortion between the subjective mean field and the observations generated by the real world remains sufficiently small.

The presence of model misspecification and subjective representation can be viewed not only as a limitation of agent decision-making, but also as a vulnerability, or, from another perspective, an arbitrage opportunity, for a third party or system designer to influence network behavior. Because agents base their decisions on internally consistent but potentially distorted representations of the world, small and well-targeted perturbations to the observations they receive can induce systematic and predictable changes in equilibrium outcomes. In particular, a designer need not alter agents' objectives or incentives directly; instead, she can shape behavior indirectly by nudging or distorting the information upon which agents condition their subjective models.

We refer to this mechanism as \textit{cognitive arbitrage}. Cognitive arbitrage exploits the gap between the true underlying environment and the agents' subjective representations, leveraging misspecification to achieve desired system-level outcomes at minimal cost. Rather than correcting agents' models, the designer strategically works within the agents' perceived world, ensuring that the induced behavior remains internally rational and observationally consistent from the agents' standpoint. We develop a principled framework for designing such distortions, characterizing when and how a designer can optimally influence equilibrium behavior while respecting explicit resource.

In this work, we investigate these ideas in the context of linear-quadratic network games, a canonical class of models that captures strategic interactions with network externalities. Linear-quadratic games are widely used due to their analytical tractability and their ability to approximate more complex interactions in applications such as opinion dynamics, economic networks, power systems, and distributed learning \citep{ballester2006s, bacsar1998dynamic}. Within this setting, we study how agents equipped with subjective, locally informed models learn and act, how equilibrium behavior emerges under model misspecification, and how a higher-level designer can strategically influence outcomes through controlled distortions. By combining tools from Berk-Nash equilibrium \citep{esponda2016bayesian}, mean-field approximations, and two-time-scale learning \citep{borkar2008stochastic}, we provide a unified framework for understanding subjective decision-making and cognitive arbitrage in networked systems.

The paper is structured as follows. Section 2 introduces the network game and the agents' subjective modeling framework. Section 3 characterizes Berk-Nash equilibrium under different classes of conjectures. Section 4 analyzes the gap between Berk-Nash and perfectly specified Nash equilibria. Section 5 formulates the cognitive arbitrage problem as a Stackelberg optimization, derives closed-form solutions, and establishes two-time-scale convergence results. Section 6 presents a numerical case study illustrating the theoretical findings. Section 7 concludes with implications and future research directions

\section{A NETWORK GAME WITH MISSPECIFIED MODELS}

We consider a network game where agents interact through unobserved externalities and learn via noisy aggregate signals \citep{jackson2015games, galeotti2010network}. We model the long-run behavior using the Berk-Nash equilibrium (BNE) concept \citep{esponda2016berk, esponda2016bayesian}.

\subsection{Network and Objective Model}
Consider a directed network $G=(V, E)$ with agents $V=\{1, \ldots, n\}$. Let $V_{i}:=\{j \in V \setminus \{i\} \mid (j, i) \in E\}$ be agent $i$'s neighborhood and $G=[g_{ij}]$ be the interaction matrix, where $g_{ij} \in \mathbb{R}$ captures the impact of $j$ on $i$ (with $g_{ii}=0$). Each agent $i$ chooses an action $x_{i} \in \mathbb{R}$ and observes a noisy signal:
\[
y_{i} = \sum_{j \in V_{i}} g_{i j} x_{j} + \eta_{i}, \quad \eta_{i} \sim \mathcal{N}(0, \sigma_{i}^{2}).
\]
Agent $i$'s true cost function is $J_{i}(x)=\frac{r_{i}}{2} x_{i}^{2}+x_{i}(y_{i}-b_{i})$, where $r_{i}>0$ represents a private cost parameter and $b_{i} \in \mathbb{R}$ is an idiosyncratic bias. While agents know their own cost parameters, they do not observe the individual actions $x_j$ ($j \neq i$) nor the weights $g_{ij}$.

\subsection{Subjective Conjectures}
Agents attribute the signal $y_i$ to a simplified subjective model $\mathcal{C}_{i}:=(z_{i}, \theta_{i})$. Specifically, agent $i$ assumes the signal follows $y_{i}=\theta_{i}^{\top} z_{i}+\eta_{i}$, where $z_{i} \in \mathbb{R}^{k}$ is a regressor representing perceived network features, and $\theta_{i} \in \Theta_{i} \subset \mathbb{R}^{k}$ is a parameter vector to be learned. We assume that while the detailed individual actions $x_j$ are unobservable, agent $i$ can directly sense the aggregated feature $z_i$ (e.g., total local activity or interference) from the environment. Common conjectures include:
\begin{itemize}
    \item Constant: $z_i = 1$. Agent assumes a static background noise.
    \item Aggregate: $z_i = \sum_{j \in V_i} x_j$. Agent perceives only total neighbor activity.
    \item Mean-field: $z_i = \frac{1}{|V_i|} \sum_{j \in V_i} x_j$. Agent reacts to average local behavior.
    \item Feature-based: $z_i=\phi(\frac{1}{|V_i|} \sum_{j \in V_i} x_j)$ for some feature map $\phi$.
\end{itemize}
Misspecification arises when the true signal generator $\sum g_{ij}x_j$ cannot be perfectly represented by $\theta_i^\top z_i$ for any $\theta_i$.

\subsection{Optimal Behavior Given Misspecified Conjectures}
Given a conjecture parameter $\theta_{i}$, agent $i$ chooses her action to minimize the subjective expected cost:
\[
\min _{x_{i} \in \mathbb{R}} \mathbb{E}_{\theta_{i}}\left[\frac{r_{i}}{2} x_{i}^{2}+x_{i}\left(y_{i}-b_{i}\right)\right].
\]
Under the subjective model, the conjectured conditional mean is $\mathbb{E}_{\theta_{i}}[y_{i}]=\theta_{i}^{\top} z_{i}$. Substituting this into the objective function yields a strictly convex optimization problem:
\[
\min _{x_{i} \in \mathbb{R}}\left\{\frac{r_{i}}{2} x_{i}^{2}+x_{i}\left(\theta_{i}^{\top} z_{i}-b_{i}\right)\right\}.
\]
The unique optimal action $x_i$ is explicitly given by the best-response map $v_{i}: \Theta_{i} \rightarrow \mathbb{R}$:
\begin{equation}
x_{i}=v_{i}(\theta_{i}):=\frac{b_{i}-\theta_{i}^{\top} z_{i}}{r_{i}}. \label{eq:best_response}
\end{equation}
Thus, the agent's action is affine in the conjectured mean signal, where $b_i$ captures intrinsic bias and $\theta_i^\top z_i$ captures the perceived strategic externalities.

\subsection{Statistical Consistency and Learning}
Let $x$ be the joint strategy profile. Under the \textit{true objective model}, agent $i$'s observation follows the distribution $P_{i}^{0}(\cdot \mid x)=\mathcal{N}(\mu_{i}^{0}(x), \sigma_{i}^{2})$, where the true mean is determined by the actual network interaction: $\mu_{i}^{0}(x)=\sum_{j \neq i} g_{i j} x_{j}.$

In contrast, under the \textit{subjective model} specified by $(z_i, \theta_i)$, the agent perceives the signal distribution as $P_{i}^{\theta_{i}}(\cdot \mid x_{i})=\mathcal{N}(\mu_{i}^{\theta_{i}}, \sigma_{i}^{2})$, where the subjective mean is $\mu_{i}^{\theta_{i}}=\theta_{i}^{\top} z_{i}$. Agents update their conjectures to fit the observations generated by equilibrium play. Following the Berk-Nash framework, a conjecture is consistent if it minimizes the Kullback-Leibler (KL) divergence between the true and subjective distributions:
\[
\theta_{i}^{*} \in \arg \min _{\theta_{i} \in \Theta_{i}} \operatorname{KL}\left(P_{i}^{0}(\cdot \mid x) \| P_{i}^{\theta_{i}}(\cdot \mid x_{i})\right).
\]
Since both distributions are Gaussian with identical variance $\sigma_i^2$, the KL divergence simplifies to the squared Euclidean distance between the means: 
\[
\operatorname{KL}\left(P_{i}^{0} \| P_{i}^{\theta_{i}}\right)=\frac{1}{2 \sigma_{i}^{2}}\left(\mu_{i}^{0}(x)-\mu_{i}^{\theta_{i}}\right)^{2}.
\]
Consequently, the statistical consistency condition reduces to a pointwise least-squares projection of the true aggregate influence onto the conjecture class:
\begin{equation}
\theta_{i}^{*} \in \arg \min _{\theta_{i} \in \Theta_{i}}\left(\sum_{j \neq i} g_{i j} x_{j}-\theta_{i}^{\top} z_{i}\right)^{2}. \label{eq:consistency}
\end{equation}

\subsection{Berk-Nash Equilibrium}
We now formally define the equilibrium concept, which requires fixed-point consistency between actions and conjectures.

\begin{definition}[Berk-Nash Equilibrium]
A Berk-Nash equilibrium (BNE) is a tuple $(x^{*}, \theta^{*})$ consisting of a strategy profile $x^{*}$ and a conjecture profile $\theta^{*}$ such that for all $i \in V$:
\begin{enumerate}
    \item \textbf{Optimality:} Given $\theta_{i}^{*}$, the action $x_{i}^{*}$ minimizes the subjective cost, satisfying \eqref{eq:best_response}: 
    \[ x_{i}^{*} = \frac{b_{i}-\theta_{i}^{* \top} z_{i}^{*}}{r_{i}}. \]
    \item \textbf{Consistency:} Given $x^{*}$, the conjecture $\theta_{i}^{*}$ minimizes the KL divergence (or equivalently, the squared error) as in \eqref{eq:consistency}: 
    \[ \theta_{i}^{*} \in \arg \min _{\theta_{i} \in \Theta_{i}}\left(\sum_{j \neq i} g_{i j} x_{j}^{*}-\theta_{i}^{\top} z_{i}^{*}\right)^{2}. \]
\end{enumerate}
\end{definition}

The BNE characterizes the steady-state where agents play optimally with respect to their misspecified models, and those models are the best statistical approximation of the environment given the agents' actions.

\section{Berk-Nash Equilibrium Analysis}

\subsection{Benchmark and Constant Conjectures}
We first consider the true Nash equilibrium (NE) as a benchmark. When agents know the interaction matrix $G=(g_{ij})$, the NE $x^{\mathrm{NE}}$ satisfies the condition $(R+G) x^{\mathrm{NE}}=b$, where $R=\operatorname{diag}(r_{i})$.

Consider now the \textit{Constant-Only Conjecture} where agents believe $y_{i}=\theta_{i}+\eta_{i}$ (i.e., $z_i \equiv 1$). The consistency condition requires $\theta_i^*$ to match the true expected signal. Thus, $\theta_{i}^{*}=\sum_{j \in V_{i}} g_{i j} x_{j}^{*}$. Substituting this learned parameter into the best response $x_{i}^{*}=(b_{i}-\theta_{i}^{*})/{r_{i}}$, we observe that the agent effectively reacts to the true aggregate interference:
\[
r_i x_i^* + \sum_{j \in V_{i}} g_{i j} x_{j}^{*} = b_i \implies (R+G) x^{*}=b.
\]
Thus, under constant conjectures, the BNE coincides with the NE. The projection onto a constant is lossless for determining equilibrium means, implying no distortion in the final outcome.

\subsection{Global Mean-Field Conjectures}
Suppose agents adopt a global mean-field regressor $z_{i} = \bar{x}_{-i} := \frac{1}{n-1} \sum_{j \neq i} x_{j}$. While the true signal is generated locally by $V_i$, agents attribute it to the global population average.
The consistent parameter $\theta_i^*$ minimizes the squared error between the true local influence and the conjectured global influence $(\sum_{j \in V_i} g_{ij}x_j - \theta_i \bar{x}_{-i})^2$. In a symmetric or large-population limit where $\bar{x}_{-i} \approx \bar{x}$, the optimal projection yields the \textit{total local influence}:
\[
\theta_{i}^{*}=\gamma_{i}:=\sum_{j \in V_{i}} g_{i j}.
\]
The resulting BNE actions are obtained by substituting $\theta_i^*$ into the best response:
\[
x_{i}^{*}=\frac{b_{i}-\gamma_{i} \bar{x}^{*}}{r_{i}}, \quad \text{with } \bar{x}^{*} = \frac{n^{-1} \sum_{i} (b_{i}/r_{i})}{1+n^{-1} \sum_{i} (\gamma_{i}/r_{i})}.
\]
Here, agents effectively simplify the complex topology into a single scalar $\gamma_i$.

\begin{proposition}[Mean-Field Limit]
Consider a sequence of networks with dense neighborhoods ($|V_i^{(n)}| \to \infty$) and scaled weights $g_{ij}^{(n)} \propto 1/n$. If row sums converge $\gamma_i^{(n)} \to \gamma$ and actions are asymptotically exchangeable, then the BNE action $x_i^{BN}$ converges almost surely to the standard Mean-Field Game (MFG) equilibrium:
\[
x^{\mathrm{MFG}}=\frac{b}{r+\gamma}.
\]
\end{proposition}
This confirms that classical MFGs appear as a limiting case of BNE where agents' global conjectures become asymptotically correct.

\subsection{Local Mean-Field and Network Sparsification}
\label{sec:LMF_analysis}
We now analyze the case where agent $i$ focuses on a subset $S_{i} \subseteq V_{i}$, using the regressor $z_{i}=|S_{i}|^{-1} \sum_{j \in S_{i}} x_{j}$.
The consistent parameter $\theta_i^*$ minimizes the Euclidean distance between the true signal and the predictor $\theta_i z_i$. This yields the ratio of true total influence to the subset average:
\begin{equation}
\theta_{i}^{*} = \frac{\sum_{j \neq i} g_{i j} x_{j}^{*}}{z_i^*} = \frac{\sum_{j \in V_i} g_{i j} x_{j}^{*}}{|S_{i}|^{-1} \sum_{j \in S_{i}} x_{j}^{*}}. \label{eq:theta_star_lmf}
\end{equation}
By substituting this $\theta_i^*$ back into the best response $x_i^* = (b_i - \theta_i^* z_i^*)/r_i$, the interaction structure is effectively re-weighted by the attention set size. This leads to a modified linear system:
\begin{equation}
(R+\widetilde{G}) x^{*}=b, \quad \text{where } \widetilde{g}_{i j}= \begin{cases}|S_{i}|^{-1} g_{i j}, & j \in S_i \\ 0, & \text{otherwise}\end{cases}. \label{eq:tilde_G}
\end{equation}
Here, $\widetilde{G}$ represents a "sparsified" perception of the network. The BNE coincides with the NE of a modified game where interaction weights are rescaled by $|S_i|^{-1}$. If $S_i$ omits payoff-relevant neighbors, the BNE systematically deviates from the true NE.

\subsection{Joint Learning Dynamics}
We define a coupled learning process where agents simultaneously update conjectures via stochastic gradient descent and actions via best response.
Let $z_i(k)$ be the regressor at step $k$. The update laws are:
\begin{align}
\theta_{i}(k+1) &= \theta_{i}(k)+\alpha_{k}(y_{i}(k)-\theta_{i}(k) z_{i}(k)) z_{i}(k), \label{eq:theta_update}\\
x_{i}(k+1) &= \frac{b_{i}}{r_{i}}-\frac{\theta_{i}(k+1)}{r_{i}} z_{i}(k). \label{eq:x_update}
\end{align}
This system represents a feedback loop where data determines conjectures, and conjectures determine actions.

\begin{theorem}[Almost-Sure Convergence]
Assume independent noise $\eta_i(k)$, diminishing step sizes $\sum \alpha_k = \infty, \sum \alpha_k^2 < \infty$, persistent excitation of $z_i(k)$, and that the spectral radius $\rho(R^{-1} \widetilde{G})<1$. Then, for any initial condition, the joint process $\{(\theta(k), x(k))\}_{k \geq 0}$ converges almost surely to the unique Berk-Nash equilibrium $(x^*, \theta^*)$.
\end{theorem}

\section{BEST-RESPONSE STRUCTURE AND VALUE OF MISSPECIFICATION}

We formalize the relationship between Nash (NE) and Berk-Nash equilibrium (BNE) via best-response mappings and quantify their divergence using the Value of Misspecification (VoM).

\subsection{Best-Response Gap}
Let $\mathcal{X}_{i}$ and $\Theta_{i}$ be agent $i$'s action and conjecture spaces, with $\mathcal{X}=\prod \mathcal{X}_i$.
The \textit{true-model best-response} $\mathrm{BR}_{i}^{0}: \mathcal{X}_{-i} \rightrightarrows \mathcal{X}_{i}$ minimizes the objective under the true distribution $P^{0}(\cdot \mid x)$:
\[
\mathrm{BR}_{i}^{0}(x_{-i}):=\arg \min _{x_{i} \in \mathcal{X}_{i}} \mathbb{E}_{P^{0}}\left[J_{i}(x_{i}, x_{-i}, y_{i})\right].
\]
A Nash equilibrium $x^{\mathrm{NE}}$ satisfies $x_{i}^{\mathrm{NE}} \in \mathrm{BR}_{i}^{0}(x_{-i}^{\mathrm{NE}}), \forall i$.

In contrast, BNE relies on the \textit{conjecture-conditioned best-response} $\mathrm{BR}_{i}^{\theta_{i}}$, which minimizes cost under the subjective distribution $P^{\theta_i}$:
\[
\mathrm{BR}_{i}^{\theta_{i}}(x_{-i}):=\arg \min _{x_{i} \in \mathcal{X}_{i}} \mathbb{E}_{P^{\theta_{i}}}\left[J_{i}(x_{i}, x_{-i}, y_{i})\right].
\]
In our LQ setting, $\mathrm{BR}_{i}^{\theta_{i}}$ yields the unique solution $v_{i}(\theta_{i})=(b_{i}-\theta_{i}^{\top} z_{i})/r_{i}$.
A BNE is a pair $(x^*, \theta^*)$ where $x^*_i \in \mathrm{BR}_{i}^{\theta_{i}^*}(x_{-i}^*)$ and $\theta_i^*$ minimizes the KL divergence between $P^0(\cdot \mid x^*)$ and $P^{\theta_i}(\cdot \mid x_i^*)$.
Thus, $x^{\mathrm{NE}}$ is a BNE if and only if there exist consistent conjectures $\theta^*$ such that $\mathrm{BR}_{i}^{\theta_{i}^*}(x_{-i}^{\mathrm{NE}}) \equiv \mathrm{BR}_{i}^{0}(x_{-i}^{\mathrm{NE}})$, i.e., the perceived marginal effects match the true externalities.

\subsection{Value of Misspecification (VoM)}
We quantify the efficiency impact of misspecified learning using the VoM metric. Let $\mathcal{J}(x):=\sum_{i} J_{i}(x) = \frac{1}{2} x^{\top} R x+x^{\top} G x-b^{\top} x$ be the aggregate cost.

\begin{definition}[Value of Misspecification]
The VoM is the relative cost deviation of the BNE $x^{\mathrm{BN}}$ from the NE $x^{\mathrm{NE}}$:
\begin{equation}
\mathrm{VoM}:=\frac{\mathcal{J}(x^{\mathrm{BN}})-\mathcal{J}(x^{\mathrm{NE}})}{\mathcal{J}(x^{\mathrm{NE}})}. \label{eq:vom_def}
\end{equation}
\end{definition}
Positive (negative) VoM implies misspecification increases (decreases) aggregate cost.

\subsection{VoM under Local Mean-Field Conjectures}
We now specialize this metric to the case of local mean-field (LMF) conjectures. Recall from Sec. \ref{sec:LMF_analysis} that the BNE actions satisfy $(R+\widetilde{G})x^{\mathrm{BN}}=b$, where $\widetilde{G}$ is the sparsified interaction matrix. The structural difference $\Delta G := \widetilde{G} - G$ drives the equilibrium gap.

By substituting the closed-form expressions for $x^{\mathrm{NE}}$ and $x^{\mathrm{BN}}$ into the quadratic cost $\mathcal{J}(x)$, we can derive explicit bounds on this metric in terms of the network distortion.
\begin{proposition}[Bounds on VoM]
Assume $R \succeq r_{\min} I$, $b \neq 0$, and that the spectral radius $\rho(G) < r_{\min}$ ensures stability. There exists a constant $C>0$ depending on the system matrices such that:
\[
|\mathrm{VoM}_{\mathrm{LMF}}| \leq C \frac{(\|R\|+\|G\|)^{2}}{r_{\min }^{2}} \|\widetilde{G}-G\|.
\]
\end{proposition}

\begin{pf}
Let $\Delta G:=\widetilde{G}-G$ denote the network distortion. Subtracting the equilibrium conditions $(R+G) x^{\mathrm{NE}}=b$ and $(R+\widetilde{G}) x^{\mathrm{BN}}=b$ yields the error dynamics:
\[
(R+G)(x^{\mathrm{BN}}-x^{\mathrm{NE}}) = -\Delta G x^{\mathrm{BN}}.
\]
Taking norms and exploiting the invertibility of $R+G$, we obtain:
\[
\|x^{\mathrm{BN}}-x^{\mathrm{NE}}\| \leq \|(R+G)^{-1}\| \|\Delta G\| \|x^{\mathrm{BN}}\|.
\]
Given $R \succeq r_{\min} I$ and $\rho(G) < r_{\min}$, the inverse is bounded by $\|(R+G)^{-1}\| \leq (r_{\min}-\rho(G))^{-1}$. Similarly, the equilibrium action is proportional to the input $b$, bounded by $\|x^{\mathrm{BN}}\| \leq \|b\|(r_{\min}-\rho(G))^{-1}$. Substituting these yields:
\begin{equation}
\|x^{\mathrm{BN}}-x^{\mathrm{NE}}\| \leq K_1 \|b\| \|\Delta G\|, \label{eq:action_bound}
\end{equation}
where $K_1 = (r_{\min}-\rho(G))^{-2}$. Note that the action deviation scales linearly with $\|b\|$.

Next, consider the aggregate cost difference. By the Mean Value Theorem, $|\mathcal{J}(x^{\mathrm{BN}})-\mathcal{J}(x^{\mathrm{NE}})| \leq \sup_{\xi} \|\nabla \mathcal{J}(\xi)\| \|x^{\mathrm{BN}}-x^{\mathrm{NE}}\|$. The gradient is $\nabla \mathcal{J}(x) = (R+G+G^\top)x - b$. Since the equilibrium actions are linear in $b$, the gradient norm along the path is bounded by $K_2 \|b\|$. Combining this with \eqref{eq:action_bound}, the cost difference scales quadratically with $b$:
\begin{equation}
|\mathcal{J}(x^{\mathrm{BN}})-\mathcal{J}(x^{\mathrm{NE}})| \leq K_3 \|b\|^2 \|\Delta G\|. \label{eq:cost_diff}
\end{equation}

To bound the relative error (VoM), we lower-bound the denominator. At Nash equilibrium, substituting $(R+G)x^{\mathrm{NE}}=b$ into the cost function yields $\mathcal{J}(x^{\mathrm{NE}}) = -\frac{1}{2} (x^{\mathrm{NE}})^{\top} R x^{\mathrm{NE}}$. Using $\|x^{\mathrm{NE}}\| \geq \|b\|(\|R\|+\|G\|)^{-1}$, we have:
\begin{equation}
|\mathcal{J}(x^{\mathrm{NE}})| \geq \frac{r_{\min}}{2} \|x^{\mathrm{NE}}\|^2 \geq K_4 \|b\|^2. \label{eq:cost_ne}
\end{equation}
Finally, dividing \eqref{eq:cost_diff} by \eqref{eq:cost_ne}, the $\|b\|^2$ terms cancel out, yielding the scale-invariant bound:
\[
|\mathrm{VoM}_{\mathrm{LMF}}| \leq \frac{K_3}{K_4} \|\Delta G\| = C' \|\Delta G\|.
\]
\end{pf}

\section{COGNITIVE ARBITRAGE}

We introduce the notion of \textit{cognitive arbitrage} to describe the deliberate design of agents' conjecture spaces in order to shape equilibrium outcomes under bounded rationality. In the Berk-Nash framework, agents do not optimize directly with respect to the true interaction structure. Instead, each agent reasons through a conjecture that maps observations into best responses. As a result, equilibrium behavior depends not only on incentives and data, but also on the expressiveness and structure of the conjectures agents use to interpret their environment.

This observation reveals a new design channel. Rather than eliminating misspecification by forcing agents to learn the full network, a designer may instead shape the conjectures through which agents reason, thereby steering the resulting Berk-Nash equilibrium toward desirable outcomes.

\subsection{The Concept of Cognitive Arbitrage}
Cognitive arbitrage refers to the strategic exploitation of this conjectural channel. The key idea is to improve equilibrium performance by manipulating the information agents use to reason about their strategic environment, thereby leveraging their misspecification without requiring full information or rationality.

\begin{definition}[Cognitive Arbitrage]
Cognitive arbitrage is the strategic manipulation of agents' equilibrium behavior by exploiting their misspecified conjectures. Specifically, it involves the \textbf{intentional injection of controlled informational distortions} into the agents' observation channels to steer the Berk-Nash equilibrium toward a desired benchmark, subject to constraints on distortion magnitude.
\end{definition}

It is crucial to note that cognitive arbitrage does not alter agents' preferences, payoff functions, or learning rules. Instead, it reshapes the \textit{projection step} inherent in Berk-Nash learning by shifting the effective data distribution onto which the agents' fixed models are projected. 

\subsection{Optimal Design via Minimal Distortion}
We now formalize this paradigm in the Linear-Quadratic (LQ) setting under local mean-field conjectures. We model the design intervention as a minimal distortion of the observation channel.

\textbf{Problem Formulation.}
Consider the LQ network game where the designer modifies the observation channel so that agent $i$ observes $\tilde{y}_{i}=\sum_{j \neq i} g_{i j} x_{j}+v_{i}+\eta_{i}$. Here, $v_i$ is a designer-induced distortion with mean $\delta_i$ and variance $\rho_i^2$.
The designer seeks to minimize the aggregate cost $\mathcal{J}(x)$ (defined in Sec. 4) subject to a budget $\Gamma$ on the distortion cost $\mathcal{C}(\delta, \rho) = \sum_{i} (\alpha_{i} \delta_{i}^{2}+\beta_{i} \rho_{i}^{2})$.

\textbf{Induced Equilibrium.}
Agents treat $\tilde{y}_i$ as the input for their learning. As derived in Sec. \ref{sec:LMF_analysis}, statistical consistency implies the induced BNE action profile $x^{\star}(\delta)$ satisfies $(R+\widetilde{G}) x^{\star}(\delta)=b-\delta$. Assuming nonsingularity, the unique response is:
\begin{equation}
x^{\star}(\delta)=M(b-\delta), \quad \text{where } M:=(R+\widetilde{G})^{-1}.
\end{equation}
Note that the variance parameter $\rho$ does not affect the mean action and is optimally set to $\rho^{\star}=0$.

\textbf{Optimization Solution.}
Substituting $x^{\star}(\delta)$ into the global objective $\mathcal{J}(x)$, the designer's problem reduces to a Quadratically Constrained Quadratic Program (QCQP):
\begin{equation}
\min _{\delta} \ f(\delta):=\delta^{\top} Q \delta-2 b^{\top} Q \delta+c^{\top} \delta \quad \text { s.t. } \quad \delta^{\top} A \delta \leq \Gamma, \label{eq:qcqp}
\end{equation}
where $Q$ is the symmetrized Hessian matrix defined by
\[
Q := \frac{1}{2} M^{\top} R M + \frac{1}{2} M^{\top} (G+G^{\top}) M,
\]
and $c:=M^{\top} b, A=\operatorname{diag}(\alpha_i)$.

\begin{theorem}[Optimal Cognitive Arbitrage Strategy]
Assume $Q \succeq 0$ and $A \succ 0$.
(i) The problem \eqref{eq:qcqp} admits a unique optimal solution $\delta^{\star}$ given by:
\[
\delta^{\star}=\left(Q+\lambda^{\star} A\right)^{-1}\left(Q b-\frac{1}{2} c\right),
\]
where $\lambda^{\star} \geq 0$ is the unique scalar satisfying the complementary slackness condition $\lambda^{\star}(\delta^{\star \top} A \delta^{\star}-\Gamma)=0$.
(ii) The resulting optimized BNE action profile is $x^{\mathrm{BN}}=M(b-\delta^{\star})$.
\end{theorem}

\begin{pf}
The optimization problem is a convex quadratically constrained quadratic program (QCQP) since $Q \succeq 0$ and $A \succ 0$. Slater's condition holds for any $\Gamma>0$, and therefore strong duality applies.

Introduce the Lagrangian with multiplier $\lambda \geq 0$:
\[
\mathcal{L}(\delta, \lambda) = \delta^{\top} Q \delta - 2 b^{\top} Q \delta + c^{\top} \delta + \lambda\left(\delta^{\top} A \delta - \Gamma\right).
\]
Since $Q$ is symmetric by definition, the stationarity condition with respect to $\delta$ is given by $\nabla_{\delta} \mathcal{L} = 2(Q+\lambda A) \delta - 2 Q b + c = 0$. This implies the explicit form:
\[
\delta(\lambda) = (Q+\lambda A)^{-1}\left(Q b - \frac{1}{2} c\right).
\]
Complementary slackness requires $\lambda(\delta^{\top} A \delta - \Gamma)=0$, subject to $\delta^{\top} A \delta \leq \Gamma$ and $\lambda \geq 0$.

If the unconstrained solution satisfies $\delta(0)^{\top} A \delta(0) < \Gamma$, the constraint is inactive and $\lambda^{\star}=0$. Otherwise, the map $\lambda \mapsto \delta(\lambda)^{\top} A \delta(\lambda)$ is strictly decreasing, and hence there exists a unique $\lambda^{\star} > 0$ satisfying $\delta(\lambda^{\star})^{\top} A \delta(\lambda^{\star}) = \Gamma$. This establishes the optimal distortion $\delta^{\star}$.

Finally, under local mean-field conjectures, the induced Berk-Nash equilibrium mapping is $x^{\star}(\delta)= M(b-\delta)$. Substituting $\delta^{\star}$ yields the equilibrium result. Uniqueness follows from the nonsingularity of $R+\widetilde{G}$.
\end{pf}

\begin{figure}[h]
\centering
\includegraphics[width=0.55\linewidth]{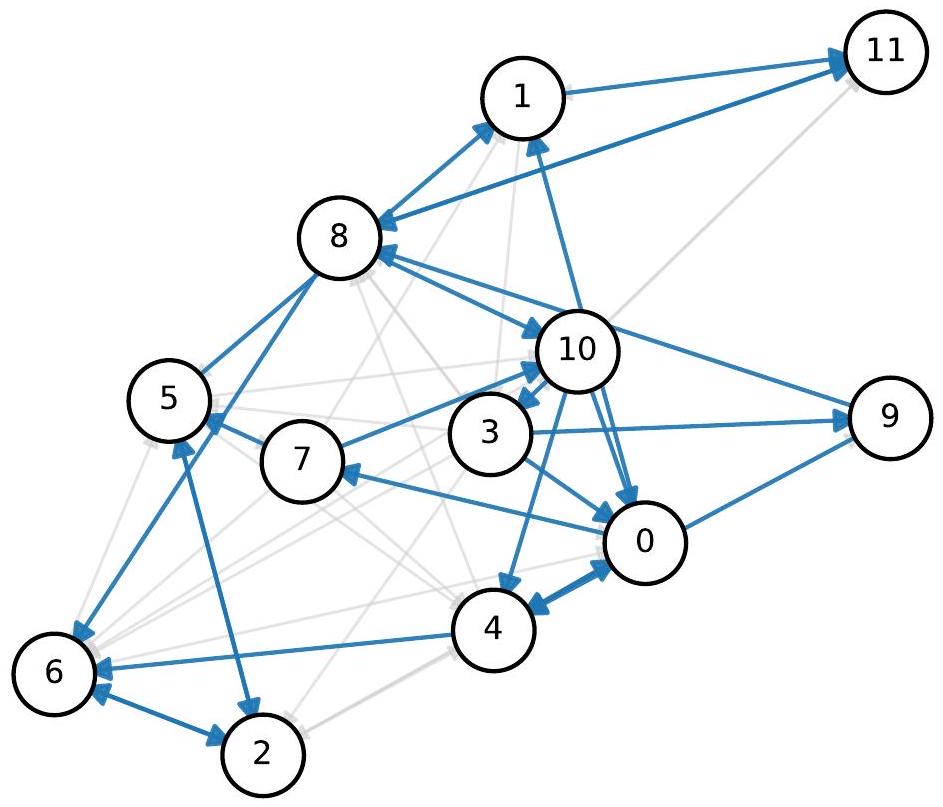}
\caption{Network mismatch: True dense interactions (gray) vs. sparse subjective attention (blue). Agents ignore long-range dependencies, creating persistent model misspecification.}
\label{fig:network}
\end{figure}

\section{NUMERICAL EXAMPLES}

We illustrate the theoretical results using a networked quadratic game with $n=12$ agents. The true interaction matrix $G$ is dense but unknown to the agents, while agents form local mean-field conjectures based on a sparse observation graph (Figure \ref{fig:network}). Each agent attends to approximately $3$ neighbors, capturing roughly $30\%$ of the total interaction weight.

\begin{figure}[h]
\centering
\includegraphics[width=1\linewidth]{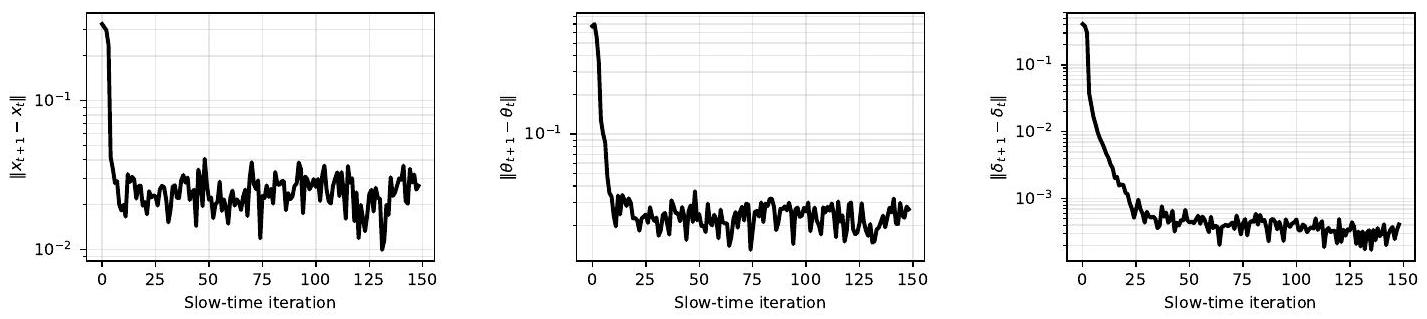}
\caption{Convergence diagnostics plotted on a log scale. The rapid decay of action and conjecture updates relative to the distortion updates ($\|\Delta \delta\|$) empirically validates the two-time-scale separation principle.}
\label{fig:convergence}
\end{figure}

\subsection{Two-Time-Scale Dynamics}
We simulate the coupled learning dynamics where agents update actions and conjectures on a fast time scale (step size $\alpha_k$), while the designer updates the distortion $\delta$ on a slow time scale (step size $\beta_k \ll \alpha_k$).
The simulation incorporates a quadratic distortion budget $\Gamma$. Figure \ref{fig:convergence} displays the convergence metric $\|\Delta \phi_k\| := \|\phi_{k+1}-\phi_k\|$ for the system variables $\phi \in \{x, \theta, \delta\}$.
The results confirm the time-scale separation: agent-level variables ($x, \theta$) decay rapidly (fast equilibration), while the designer's distortion $\delta$ evolves slowly, validating the quasi-static assumption essential for the Stackelberg approximation.

\begin{table}[h]
\centering
\small 
\setlength{\tabcolsep}{3pt} 
\caption{Equilibrium values: distortion ($\delta^{\star}$), action ($x^{\star}$), and conjecture ($\theta^{\star}$).}
\begin{tabular}{c|ccc||c|ccc}
\hline
Agent & $\delta_{i}^{\star}$ & $x_{i}^{\star}$ & $\theta_{i}^{\star}$ & Agent & $\delta_{i}^{\star}$ & $x_{i}^{\star}$ & $\theta_{i}^{\star}$ \\
\hline
0 & 0.39 & 1.03 & 0.00 & 6 & -0.16 & -0.62 & 0.77 \\
1 & -0.10 & -0.49 & 0.35 & 7 & 0.43 & 1.52 & -0.16 \\
2 & 0.65 & 1.58 & -1.13 & 8 & 0.10 & 0.09 & 0.15 \\
3 & 0.30 & 1.00 & -0.19 & 9 & 0.86 & 2.88 & -0.71 \\
4 & 0.39 & 0.82 & 0.24 & 10 & 0.55 & 1.58 & -0.28 \\
5 & 0.41 & 1.37 & -0.08 & 11 & 0.23 & 1.09 & 0.55 \\
\hline
\end{tabular}
\label{tab:equilibrium}
\end{table}

\subsection{Equilibrium Analysis}
The system converges to a unique BN Stackelberg equilibrium, detailed in Table \ref{tab:equilibrium}.
The optimal distortions $\delta^{\star}$ exhibit significant heterogeneity, reflecting a strategy to target ``central'' agents (e.g., Agent 9, who receives the maximal distortion $\delta_9^{\star} \approx 0.86$).
This targeted manipulation yields a lower aggregate cost compared to the baseline ($\delta=0$), validating the efficacy of cognitive arbitrage.

\section{Conclusion}

This paper presented a framework for analyzing and designing network games with misspecified models. We showed that when agents rely on simplified conjectures—such as local mean-field approximations—the resulting Berk-Nash equilibrium can systematically deviate from the true Nash equilibrium. This deviation, quantified by the Value of Misspecification, reveals that subjective learning is not merely a limitation but a design channel. We introduced cognitive arbitrage as a mechanism to exploit this channel, allowing a designer to optimize system performance by minimally distorting the information agents use to form their world models.

Our analysis of linear-quadratic games provided closed-form characterizations of both the BNE and the optimal cognitive distortion strategy. Numerical results validated the theoretical findings, demonstrating the efficacy of the proposed two-time-scale learning dynamics. Future work will extend this framework to dynamic games with Markovian states, explore the robustness of cognitive arbitrage against adversarial agents, and investigate decentralized mechanisms where agents endogenously select their conjecture classes.

\bibliography{ifacconf}

\end{document}